\documentclass[12pt]{article}

\usepackage[margin=1.0in]{geometry}
 
\usepackage{graphicx, subcaption}      
\usepackage{amssymb, amsmath}
\usepackage{tikz}
\usetikzlibrary{patterns,angles,quotes,arrows}
\usepackage[shortlabels]{enumitem}
\usepackage{amsthm}
\usepackage[numbers]{natbib}
\newtheorem{assumption}{Assumption}
\newtheorem{theorem}{Theorem}
\newtheorem{lemma}{Lemma}
\newtheorem{definition}{Definition}
\newtheorem{problemStatement}{Problem Statement}
\newtheorem{proposition}{Proposition}
\newtheorem{remark}{Remark}

\usepackage{tikz}
\usetikzlibrary{patterns,angles,quotes,arrows}
\newcommand{\parenthnewln}{\right.\left.\quad\quad{}}
\usepackage{pbox}

\usepackage{lipsum}

\providecommand{\keywords}[1]{\textbf{\textit{Index terms---}} #1}
\usepackage{cite}

\def\BibTeX{{\rm B\kern-.05em{\sc i\kern-.025em b}\kern-.08em
    T\kern-.1667em\lower.7ex\hbox{E}\kern-.125emX}}
\begin{document}

\title{Guaranteed Reachability on Riemannian Manifolds for Unknown Nonlinear Systems}
\author{Taha Shafa and
	Melkior Ornik
	\thanks{%
		Submitted for review on April 10, 2024. This work was supported in part by NASA grant 80NSSC22M0070 and the Air Force Office of Scientific Research under award number FA9550-23-1-0131.}%
	\thanks{Taha Shafa and Melkior Ornik are with the Department of Aerospace Engineering and the Coordinated Science Laboratory, University of Illinois Urbana-Champaign, Urbana, USA (e-mail: tahaas2@illinois.edu, mornik@illinois.edu).}
}

\maketitle

\begin{abstract}
Determining the reachable set for a given nonlinear system is critically important for autonomous trajectory planning for reach-avoid applications and safety critical scenarios. Providing the reachable set is generally impossible when the dynamics are unknown, so we calculate underapproximations of such sets using local dynamics at a single point and bounds on the rate of change of the dynamics determined from known physical laws. Motivated by scenarios where an adverse event causes an abrupt change in the dynamics, we attempt to determine a provably reachable set of states without knowledge of the dynamics. This paper considers systems which are known to operate on a manifold. Underapproximations are calculated by utilizing the aforementioned knowledge to derive a guaranteed set of velocities on the tangent bundle of a complete Riemannian manifold that can be reached within a finite time horizon. We then interpret said set as a control system; the trajectories of this control system provide us with a guaranteed set of reachable states the unknown system can reach within a given time. The results are general enough to apply on systems that operate on any complete Riemannian manifold. To illustrate the practical implementation of our results, we apply our algorithm to a model of a pendulum operating on a sphere and a three-dimensional rotational system which lives on the abstract set of special orthogonal matrices.
\end{abstract}


\section{Introduction}

Control algorithms have classically required models to manipulate system behavior, however to deal with scenarios where such models are not available, active controls research often aims to extend the capabilities of automation by developing sophisticated control algorithms under modeling uncertainties using some nominal representation of the dynamics. We go a step beyond to construct certifiably attainable control capabilities in real time for systems without knowledge of the system dynamics.

Work presented in \citep{shafa2022reachability, shafa2022maximal} sits within this general area, but those papers are unable to handle control systems that operate outside Euclidean space \citep{jurdjevic1972control,altafini2007feedback,zhu2017geometric}, which are becoming increasingly more important as state-of-the-art autonomous research involves adapting geometric methods for control of such systems \citep{lee2010geometric,chirikjian2011stochastic,schattler2012geometric,liu2019affine}. We thus extend results to systems operating on general Riemannian manifolds \citep{jost2008riemannian,bishop2012tensor} to incorporate this large class of systems in our novel theory. Motivated by examples such as a satellite becoming damaged in orbit \citep{benninghoff2014autonomous}, legged robotics applications \citep{fan2020geometric}, or other autonomous systems operating on manifolds experiencing abrupt changes in their dynamics \citep{spong1987integral,kilin2015spherical}, we aim to underapproximate the unknown system's set of reachable states \citep{brockett1976nonlinear,isidori2013nonlinear} with limited information about the dynamics. We call such a set the \textit{guaranteed reachable set} (GRS).

Unlike much of the previous work in reachable set analysis, the results presented in this paper focus on \textit{guaranteed} reachability, i.e., states that are \textit{provably} reachable, whereas much of previous work \citep{mitchell2003overapproximating,kong2018reachable} focuses on \textit{overapproximations}, i.e., states that are \textit{optimistically} reachable. We do so to develop the autonomous capability for a system to determine what it can provably achieve without knowledge of its dynamics. The primary contribution of this paper is to produce a meaningful underapproximation of the GRS for a nonlinear control-affine system operating on any complete Riemannian manifold. We do so without knowing almost anything about the system dynamics, only assuming knowledge of: (i) the local dynamics at some initial state possibly determined using persistent excitation \citep{ornik2019control}, (ii) Lipschitz bounds on the growth rate of the aforementioned local dynamics gathered using any prior knowledge of the system design and known physical laws, (iii) the Riemannian manifold and corresponding metric tensor on which the unknown system operates. 

Our approach relies on the interpretation of the unknown nonlinear control system as an \textit{ordinary differential inclusion} \citep{aubin2012differential, smirnov2002introduction} (ODI) whose right-hand side is equal to the set of velocities the control system can achieve from any initial state on a Riemannian manifold. For an unknown system, exact velocities cannot be calculated from almost any state on a manifold, but we can determine a family of achievable velocity sets from any state using the local dynamics and Riemannian Lipschitz \citep{canary2006fundamentals} bounds on the growth rate of said dynamics. The intersection of all such sets forms the \textit{guaranteed velocity set} (GVS). We underapproximate the GVS with a set characterized by a ball whose simple geometric properties allow for feasible real-time calculation.

Previous reachability research under uncertainties requires significantly more information to implement than the method presented in this paper. Examples include computing reachable sets for dynamics with bounded disturbances \citep{ding2011reachability,rakovic2006reachability} or parametric uncertainties \citep{mohan2016convex,althoff2008reachability}. More classical approaches of adaptive and robust control have applied similar methods without calculating reachable sets to achieve control objectives for uncertain systems on manifolds \citep{wingo2020adaptively,calinon2020gaussians}, however, such methods assume more information regarding the system dynamics and lack certifiably achievable guarantees from reachability analysis. Previous work focusing on reachability analysis on manifolds \citep{sussmann1987reachability,ayala2017continuity} also assumes complete or partial knowledge of the system dynamics. In contrast, our method requires no significant knowledge in the sense of a nominal model of the true dynamics of the unknown nonlinear control-affine system. Recent examples of data-driven framework implement neural network-based models to approximate reachable sets \citep{thapliyal2023approximating} and Monte Carlo methods to produce probabilistic reachable sets \citep{devonport2020data} but produce overapproximations and avoid producing any results in non-Euclidean space.

The outline of this paper is as follows: we discuss the problem in greater detail in Section II, providing important definitions needed to understand results in subsequent sections. In Section III, we determine an underapproximation of the GVS as a subset of the tangent space for any existing point on a manifold. In Section IV, we determine how the GVS can be utilized to calculate the GRS. We then illustrate by example the theory working a system operating on $SO(3)$ in Section~V. 

\textbf{Notation:}
Throughout this paper, we adopt the following notations. For notations associated with differential geometric terminologies, we refer the reader to \citep{do1992riemannian,jost2008riemannian} for additional background. $\mathbb{R}^{n\times m}$: the set of all $n\times m$ real matrices; $[n]:$ the set $\{1,\hdots,n\}$ where $n \in \mathbb{N}$; $N^T \in \mathbb{R}^{m\times n}$: the transpose of $N~\in~\mathbb{R}^{n\times m}$; $\lambda_i(N)$: the $i$-th eigenvalue of $N$ in descending order; $\sigma_i(N)$: the $i$-th singular value of $N$ in descending order; $N^\dagger$: the Moore-Penrose pseudoinverse of $N$; $\|\cdot\|$: the Euclidean norm on $\mathbb{R}^n$; $\|\cdot\|_1$: the one-norm on $\mathbb{R}^n$; $\mathbb{B}^n(a;b)$: a closed ball in $\mathbb{R}^n$ centered at $a \in \mathbb{R}^n$ with radius $b \geq 0$ under the Euclidean norm; $\mathrm{Tr}(N)$: the trace $N~\in~\mathbb{R}^{n\times m}$; $\mathrm{rank}(N)$: the rank of $N$; $\mathrm{Im}(N)$: the image (range space) of~$N$; $\mathrm{Ker}(N)$: the kernel (null space) of $N$; $(M^n,h)$: an $n$-dimensional Riemannian manifold with a Riemannian metric $h$; $\Gamma(M)$: the space of smooth vector fields on a manifold $M$; $T_xM$: the tangent space at $x \in M$; $TM$: the tangent bundle of $M$; $h_x(u,v)$: the Riemannian inner product of $u,\,v \in T_xM$; $|\cdot|_{h_x}$: the norm induced by the Riemannian inner product $h_x$; $\nabla$: the Levi-Civita connection; $\Tilde{\nabla}$: the flat connection; $\mathcal{L}(c)$: the length of curve $c$; $d(x,y)$: the Riemannian distance between $x$ and $y$; $\tau_p^qV: $ the parallel transport of vectors $v_i \in \Gamma(M)$ where $V = \begin{bmatrix}v_1 & v_2 & \hdots \end{bmatrix}$ from $p$ to $q$, i.e., $\begin{bmatrix}\tau_p^q v_1 & \tau_p^q v_2 & \hdots\end{bmatrix}$.
\section{Problem Statement}

We consider the nonlinear control-affine system $\mathcal{M}(f,G)$ on a manifold $M$ \citep{nijmeijer1990nonlinear,abraham2012manifolds,frankel2011geometry} of the form 

\begin{equation}\label{nonlinearSystem}
\begin{gathered}
    \dot{x}(t) = f(x(t)) + G(x(t))u(t)\\ = f(x(t)) + \sum_{l=1}^mg_l(x(t))u_l(t),\quad x(0) = x_0 \in M
\end{gathered}
\end{equation}

\noindent where the state space $(M^n,h)$ is a complete, connected real $n$-dimensional analytic manifold, which is consistent with standard assumptions from literature \citep{jost2008riemannian} regarding smooth manifolds. The manifold comes equipped with the Riemannian metric $h$; additionally, $x:~[0,\infty) \to M$ and $f,~g_l \in \Gamma(M)$ are globally Riemannian Lipschitz continuous with Riemannian Lipschitz constants $L_f \geq 0$, $L_{g_l} \geq 0$, and $u(t)~\in~\mathbb{B}^m(0;1) = \mathcal{U}$. The notion of a Riemannian Lipschitz function will be formally defined in subsequent sections. We begin by formally defining our first assumption, which is largely consistent with Assumption $1$ of \citep{shafa2022reachability} with the added generalization of \eqref{nonlinearSystem} operating on any complete Riemannian manifold.

\begin{assumption}\label{Assumption 1}
    System $\mathcal{M}(f,G)$ lies on a complete $n$-dimensional Riemannian manifold with a known atlas. We denote $h$ as the corresponding Riemannian metric. Functions $f$ and $G$ are of the form $f(x) = Rr(x)$ and $G(x) = RQ(x)$ where $R:\mathbb{R}^m\to T_xM$ and is constant with respect to the bases of $T_xM$ across all $x \in M$ and $r:~M \to \mathbb{R}^m$, $Q:~M \to \mathbb{R}^{m\times m}$ are functions such that $Q(x_0)$ is invertible. 
\end{assumption}

Manifolds come equipped with smooth atlases $\{(U_\alpha,\psi_\alpha)\}_{\alpha \in \mathcal{A}}$ comprised of charts \citep{do1992riemannian,jost2008riemannian} $(U_\alpha,\psi_\alpha)$ such that $U_\alpha$ are open, $\cup_{\alpha}U_\alpha = M$ and $\psi_\alpha~:~M\to\mathbb{R}^n$. For any coordinate system $(x)$ which holds $p \in U_\alpha$, the $n$ vectors $\left.\frac{\partial}{\partial x^1}\right|_p,\hdots,\left.\frac{\partial}{\partial x^n}\right|_p$ form a basis for the $n$-dimensional tangent space $T_pM$ \citep{frankel2011geometry}. Hence, the basis for the tangent space is dependent on the choice of coordinates for $U_\alpha$ and the point $p$, which we know by Assumption \ref{Assumption 1}. Calculations can span multiple charts so long as the local coordinates are adjusted accordingly. 

Assumption \ref{Assumption 1} implies $\mathrm{Im}(f(x)) \subset \mathrm{Im}(G(x))$ with the case of full actuation, i.e., $m = n$, corresponding to $R = I$. We note that $T_xM$ is diffeomorphic to $\mathbb{R}^n$, so the image of matrices composed of vectors in $T_xM$ and $\mathbb{R}^n$ have equivalent definitions as after an implicit mapping between the two spaces. The Riemannian metric in Assumption \ref{Assumption 1} provides us with knowledge of the intrinsic geometry. Tangent spaces for all $x \in M$ may contain different bases, and their connection defines how the basis changes from $T_{x_0}M$ to $T_xM$. When we refer to $R$ as constant, we imply it remains constant with respect to the connection which defines the change of basis from $T_{x_0}M$ to $T_xM$. We will formally define a connection on a manifold in subsequent sections. Similar to previous work \citep{ornik2019control,shafa2022reachability,shafa2022maximal}, we also assume knowledge on the local dynamics and growth rate bounds of system \eqref{nonlinearSystem}.

\begin{assumption}\label{Assumption 2}
    The local dynamics $f(x_0)$ and $g_l(x_0)$ are known as well as Riemannian Lipschitz bounds $L_f,\,L_{g_l}~\in~\mathbb{R}_+$. We also assume knowledge of $\mathrm{Im}(R)$. However, we do not need to know the value of $R$ exactly.
\end{assumption}

Through previous work \citep{ornik2019control}, we can calculate $f(x_0)$, $g_l(x_0)$ with arbitrarily small error using historical trajectory data. Using known physical laws, we may also determine some bound on the rate of change of the dynamics as the system travels along some manifold; to quantify this bound, we need to go over some geometric preliminaries to define the notion of parallel transport and the covariant derivative \citep{do1992riemannian,jost2008riemannian}. We then show how parallel transport is used to define the notion of Lipschitzness on curved manifolds. 

\subsection{Preliminaries}


There are a number of key differences to consider when applying guaranteed reachability theory on systems operating on $\mathbb{R}^n$ versus systems operating on $M$, namely the notion of distance, how to measure velocities for unknown systems, and what trajectories on $M$ satisfy these velocities for some given time horizon. We begin by defining a notion of distance on~$M$ by formally introducing the Riemannian inner product --- also referred to as the Riemannian metric.

\begin{definition}[Riemannian inner product]\label{RiemannianMetric}
    Let $M$ be a smooth manifold. Then, $h$ is a smooth Riemannian inner product if for any $v_1,v_2 \in T_xM$, $h_x(v_1,v_2) = v_1^TH_xv_2$ where $H_x$ is symmetric, positive definite and depends smoothly on $x \in M$.
\end{definition}

Following this definition, we see that the norm induced by the inner product on $T_xM$ is then 

\vskip -5pt

\begin{equation}\label{riemannianNorm}
    |v|_{h_x} = \sqrt{h_x(v,v)} = \sqrt{v^TH_xv}.
\end{equation}

\noindent We also denote $H_x$ as the \textbf{Riemannian metric tensor}. In the right-hand side of \eqref{riemannianNorm}, $H_x$ is a symmetric, positive definite matrix at $x~\in~M$ whose representation depends on the basis vectors of $T_xM$. Note that as a consequence of positive definiteness, $H_x$ is invertible \citep{strang2016introduction}. In Euclidean space, $H_x~=~I$ and \eqref{riemannianNorm} becomes the Euclidean norm. Defining such a Riemannian metric $h$ is always locally possible \citep{jost2008riemannian}. 

We now use the standard definition for vector induced norms on matrices \citep{strang2016introduction} and the Riemannian metric to define the Riemannian matrix norm. Let $v \in T_xM$ and $A \in \mathbb{R}^{n\times n}$. Then,

\vskip -10pt

\begin{equation}\label{riemannianMatrixNorm}
    |A|_{h_x} = \sup_{v\neq 0}\frac{|Av|_{h_x}}{|v|_{h_x}}.
\end{equation}

Note that in the case where $A \in \mathbb{R}^{n\times m}$ and $n > m$, we can instead consider taking the norm of $\bar{A} \in \mathbb{R}^{n\times n}$ where the remaining $n-m$ columns are $0$ since it would produce the same result. For notational ease, we do not differentiate notation in the case where $A \in \mathbb{R}^{n\times m}$ in subsequent sections. Using \eqref{riemannianNorm}, \eqref{riemannianMatrixNorm} we calculate norm equivalence relations between the Riemannian and Euclidean norms for subsequent calculations.

\begin{lemma}\label{Lemma_normEquivalenceRelations}
    Let $v,\,a_i,\,A,$ and $H_x$ be defined as above. Then 
    \begin{equation}\label{vectorNormEquivalency}
        \|H^{-1}_x\|^{-\frac{1}{2}}\|v\| \leq |v|_{h_x} \leq \|H_x\|^{\frac{1}{2}}\|v\|
    \end{equation}
    and 
    \begin{equation}\label{matrixNormEquivalency}
        \left(\frac{\|H_x^{-1}\|^{-1}}{\|H_x\|}\right)^{\frac{1}{2}}\|A\| \leq |A|_{h_x} \leq \|A\|\left(\frac{\|H_x\|}{\|H_x^{-1}\|^{-1}}\right)^{\frac{1}{2}}.
    \end{equation}
\end{lemma}

To calculate an underapproximation of the reachable set, we will need to measure the length of a curve on a manifold. We begin by introducing the notion of a \textbf{simple curve}: a map $\gamma:~[a,b]\to M$ is simple if it is an image of an injective, continuous map, i.e., a non-self-intersecting loop on a finite subset of its domain \citep{jost2008riemannian, sulovsk2012depth}. The simple curve along with Riemannian metric can be used to define the Riemannian length functional and distance.

\begin{definition}[Riemannian length functional]\label{RiemannianLength}
    Let $\gamma:~[a,b] \to M$ be a simple curve. We define its length as $$\mathcal{L}(\gamma) := \int_a^b\sqrt{h_{\gamma(t)}(\dot{\gamma}(t),\dot{\gamma}(t))}~dt = \int_a^b|\dot{\gamma}(t)|_{h_{\gamma(t)}}~dt.$$
\end{definition}

The length $\mathcal{L}(\gamma)$ of two curves tracing the same path is agnostic under time reparameterizations \citep{jost2008riemannian}. The shortest length induced by the Riemannian metric is the Riemannian distance. 

\begin{definition}[Riemannian distance]\label{RiemannianDistance}
    Denote $C(x,y)$ be the set of \textit{differentiable simple curves} $\gamma: [a,b] \to M$ with $\gamma(a)~=~x$ and $\gamma(b)~=~y$. We define the Riemannian distance function $d:~M\times~M~\to~[0,\infty)$ between two points $x,\,y\in M$ by $$d(x,y) = \inf_{\gamma\in C(x,y)} \mathcal{L}(\gamma).$$
\end{definition}

We can now use the notion of the Riemannian distance to define Lipschitz continuity on Riemannian manifolds. The corresponding Lipschitz constants can be used to determine maximal growth rate bounds for $f,\,g_l \in \Gamma(M)$. In subsequent sections, we use these maximal growth rate bounds on the dynamics of \eqref{nonlinearSystem} to determine underapproximations of the guaranteed set of reachable states.

\subsection{Lipschitz Continuity on Riemannian Manifold}

Recall that in $\mathbb{R}^n$, $f:\mathbb{R}^n \to \mathbb{R}^n$ is said to be Lipschitz continuous if there exists a constant $L > 0$ such that $$\|f(x) - f(y)\| \leq L\|x - y\|$$ for any $x,\,y \in \mathbb{R}^n$. However, since $T_xM \neq T_yM$, $f(x) - f(y)$ is undefined. We thus consider the Riemannian version of Lipschitz continuity found in Chapter II.A of \citep{canary2006fundamentals}. We first need to define a map between tangent spaces. Let $\vec{e}_i(x) \in T_xM$ be basis vectors known by Assumption \ref{Assumption 1}. Using the notion of an \textbf{affine connection} \citep[Section 9.1]{frankel2011geometry} and \textbf{covariant derivative} \citep{o2006elementary} along some curve $\gamma(t)$, we can write $f(x) = \sum_{j=1}^nf^j(x)\vec{e}_j(x)$ and apply the chain rule: $$
    \nabla_{\dot{\gamma}}f(x) = \sum_{i,j = 1}^n\dot{\gamma}^i\frac{\partial f^j(x)}{\partial x^i}\vec{e}_j + \sum_{j=1}^nf^j(x)\nabla_{\dot{\gamma}}\vec{e}_j.
$$Incorporating \textbf{connection coefficients} \citep{do1992riemannian,jost2008riemannian} --- also called Christoffel symbols --- defined by $\nabla_{\vec{e}_i}\vec{e}_j = \sum_{k=1}^n\Gamma_{ij}^k\vec{e}_k$ gives 

\begin{equation}\label{covariantDerivative_Definition}
    \nabla_{\dot{\gamma}}f(x) = \sum_{i,j,k = 1}^n\left(\dot{\gamma}^i\left(\frac{\partial f^k}{\partial x^i} + f^j\Gamma_{ij}^k\right)\right)\vec{e}_k.
\end{equation}

These coefficients quantify how the covariant derivative changes as the basis vectors of tangent planes on $M$ rotate along a path on a curved surface. We can define different connections by defining their corresponding connection coefficients $\Gamma_{ij}^k$. For the purposes of this paper, we will focus on two connections, namely the well-known Levi-Civita connection \citep{do1992riemannian,jost2008riemannian} and the flat connection. The Levi-Civita connection is commonly used in Riemannian geometry for its \textbf{torsion-free} and \textbf{metric-compatible} \citep{do1992riemannian,jost2008riemannian} properties. The corresponding connection coefficients are defined as 

\begin{equation}\label{connectionCoefficients}
    \Gamma_{ij}^k = \frac{1}{2}\sum_lH^{kl}\left(\frac{\partial H_{li}}{\partial x^j} + \frac{\partial H_{lj}}{\partial x^i} - \frac{\partial H_{ij}}{\partial x^l}\right).
\end{equation}

\noindent By Assumption \ref{Assumption 1}, the Riemannian metric $h$ is fully known and the knowledge of all coefficients follows from knowing $h$. In \eqref{connectionCoefficients}, $H_{kl}$ and $H^{kl}$ are components of $H_x$ and $H_x^{-1}$ respectively. With a connection formally defined, we refer the reader to \citep[Definition 4.1.2]{jost2008riemannian} for a definition of \textbf{parallel transport}.



From the definition above, the parallel transport along some curve $\gamma$ is dependent on the defined connection. As discussed previously, we will consider two connections, namely the Levi-Civita and flat connections. While the Levi-Civita connection coefficient from \eqref{connectionCoefficients} defines the parallel transport, the flat connection is the case where $\Gamma_{ij}^k = 0$ uniformly. Naturally, the \textbf{flat transport} is calculated using Definition 4.1.2 from \citep{jost2008riemannian} using $\Gamma_{ij}^k = 0$. From its corresponding metric tensor $H_x = I$, we see from \eqref{riemannianNorm} that $|\cdot|_{h_x} = \|\cdot\|$, so the flat transport is metric compatible with the Euclidean metric. We use the parallel transport to begin formally defining Riemannian Lepschitz continuity.

\begin{definition}[Classical Lipschitz Constant]\label{riemannianLipschitzDefinition}
    Let $V$ be a continuous vector field on $M$. Then $L$ is the \textit{classical Lipschitz constant} on $V$ if $$L = \sup_\gamma\frac{|\tau_\gamma V(\gamma(0)) - V(\gamma(1))|_{h_x}}{\mathcal{L}(\gamma)}$$ where $\gamma:[0,1]\to M$ varies over all $\mathcal{C}^1$-paths and $\tau_\gamma$ is shorthand for the parallel transport along the curve $\gamma$ from $\gamma(0)$ to $\gamma(1)$.
\end{definition}

By the Hopf-Rinow Theorem \citep[Theorem 1.7.1]{jost2008riemannian}, closed and bounded subsets of $M$ are compact. It follows from compactness that if $V$ is $\mathcal{C}^1$, then $L$ is finite \citep{canary2006fundamentals}. Another consequence of compactness is to find the Lipschitz constant in Definition \ref{riemannianLipschitzDefinition}, we need not vary over all $\mathcal{C}^1$-paths.

\begin{lemma}[Lemma II.A.2.4, \citep{canary2006fundamentals}]\label{Lemma 2}
    Let $U \subseteq M$ be compact and $V:U\to TM$. Then, the supremum which determines the Lipschitz constant is attained if we vary only over geodesic paths in $M$.
\end{lemma}

For a complete Riemannian manifold, given some appropriate neighborhood $U$, there is a guaranteed existence of some unique geodesic curve $\gamma:~[0,1] \to M$ connecting any two points $x,y \in M$ \citep{jost2008riemannian}. It follows that since we will carry out calculations on a domain to be defined later, the Lipschitz constants need only be calculated for a single geodesic path. Thus, Assumption~\ref{Assumption 2} is reasonable; if we carry out our calculations in some neighborhood of $x_0$, this implies there is no need to vary over infinite $\mathcal{C}^1$ paths to find $L_f$ and $L_{g_l}$. We can now express the Riemannian Lipschitz bounds introduced in Assumption \ref{Assumption 2} using the Riemannian metric $|\cdot|_{h_x}$ and Riemannian distance function $d(\cdot,\cdot)$. Let $U \subseteq M$ be compact and assume system \eqref{nonlinearSystem} is traveling along some path $\gamma:~[0,1]\to \Gamma(U)$ joining two points $p,q\in U$; if $$|\tau_p^q\Gamma(p) - \Gamma(q)|_{h_x} \leq Ld(p,q),$$ then any path $\gamma$ joining $p$ and $q$ is $L$-path Lipschitz \citep{canary2006fundamentals}.

We know Riemannian Lipschitz bounds on the unknown dynamics, namely we are given the right-hand sides of inequalities $$|\tau_{x_0}^xf(x_0) - f(x)|_{h_x} \leq L_fd(x_0,x)$$ and $$|\tau_{x_0}^xg_l(x_0) - g_l(x)|_{h_x} \leq L_{g_l}d(x_0,x)$$ for every $l \in [m]$. Given we work with a series of vectors transported along a Riemannian manifold, it would be helpful to characterize a bound on the set of vectors undergoing parallel transport simultaneously.

\begin{lemma}\label{Lemma_LGBound}
    Let $L_{g_l},\,g_l(x_0)$, $n$, and $H_x$ be as defined above. Define $\overline{L}_{g} = \max_l\{L_{g_l}\}$. If $L_G = n\|H_x^{-1}\|\|H_x\|^{\frac{1}{2}}\overline{L}_{g}$, then

    \begin{equation}\label{LGBound}
        |\tau_{x_0}^xG(x_0) - G(x)|_{h_x} \leq L_Gd(x_0,x).
    \end{equation}
\end{lemma}

The lemma above leverages norm equivalence relations to get a Riemannian Lipschitz bound $L_G$ on the matrix-valued function $G(x_0)$ composed of columns of vectors $g_l(x_0) \in T_{x_0}M$ for any $x_0 \in M$. From Lemma \ref{Lemma_LGBound} and Assumption \ref{Assumption 2}, we know $L_G$ and thus also know the right-hand side of the inequality $$|\tau_{x_0}^xG(x_0) - G(x)|_{h_x} \leq L_Gd(x_0,x).$$

Let us now denote a set $\Gamma_{L_f}(M)$ as the set of all $\hat{f} \in \Gamma(M)$ such that $|\tau_p^q\hat{f}(p) - \hat{f}(q)|_{h_x} \leq L_fd(p,q)$ for all $p,\,q \in U \subset M$. Next, we denote $\mathcal{D}_{con} \subseteq \Gamma_{L_f}(M) \times \Gamma_{L_G}(M)$ as the set of all Lipschitz $\hat{f},\hat{g}_l \in \Gamma(M)$ that are consistent with Assumptions \ref{Assumption 1} and \ref{Assumption 2}. We now utilize the set $\mathcal{D}_{con}$ to define the \textit{guaranteed velocity set} (GVS).

\subsection{Guaranteed Velocity Set}

We begin by following the approach of interpreting ordinary differential equations with control inputs as inclusions \citep{aubin2012differential,smirnov2002introduction} as in previous work \citep{shafa2022reachability, shafa2022maximal}. We will later show how to integrate the sets on the right-hand side of the inclusion over a manifold to arrive at a set of guaranteed reachable states. 

We adopt the same definitions for the \textit{available velocity set} and \textit{guaranteed velocity set} as in \citep{shafa2022reachability} with the only difference being how we defined $\mathcal{D}_{con}$ for systems operating on Riemannian manifolds in the previous section. We define the \textit{available velocity set} of system \eqref{nonlinearSystem} at state $x$ by $\mathcal{V}_x = f(x) + G(x)\mathcal{U}$ and introduce the following ODI: 

\begin{equation}\label{availableVelocitySet}
    \dot{x} \in \mathcal{V}_x = f(x) + G(x)\mathcal{U},\quad x(0) = x_0.
\end{equation}

\noindent We now define the \textit{guaranteed velocity set} (GVS) below.

\begin{definition}[Guaranteed velocity set]\label{guaranteedVelocitySetDefinition}
    Let $(\hat{f},\hat{G}) \in \mathcal{D}_{con}$ be consistent with knowledge from Assumption \ref{Assumption 2}. Then, the \textit{guaranteed velocity set} is 
\begin{equation}\label{guaranteedVelocitySet}
    \mathcal{V}^\mathcal{G}_x = \bigcap_{(\hat{f},\hat{G}) \hspace{1 mm} \in \hspace{1 mm} \mathcal{D}_{con}}
    \hat{f}(x) + \hat{G}(x)\mathcal{U} \subseteq \mathcal{V}_x.
\end{equation}
\end{definition}

The GVS is the set of all velocities that can be taken by all systems given Assumption \ref{Assumption 2}. By Assumption \ref{Assumption 2}, $\mathcal{V}^\mathcal{G}_{x_0}$ is known; using $\mathcal{V}_{x_0}$ and $\mathcal{D}_{con}$ we will develop guaranteed underapproximations for $\mathcal{V}^\mathcal{G}_x$.

Clearly, if a trajectory satisfies the inclusion $\dot{x} \in \mathcal{V}^\mathcal{G}_x$, it satisfies \eqref{availableVelocitySet} and it serves as a solution to the control system \eqref{nonlinearSystem} for an admissible control input. Using this fact, we calculate a \textit{guaranteed reachable set} from an underapproximation of $\mathcal{V}^\mathcal{G}_x$. 

\subsection{Guaranteed Reachable Set}

We ultimately want to underapproximate the set of reachable states using solely the knowledge of $\mathcal{D}_{con}$. We first define the \textit{(forward) reachable set}.

\begin{definition}[Forward reachable set]\label{forwardReachableSetDefinition}
    Let $\phi^{\hat{f},\hat{G}}_u(\cdot; x_0)$ denote the controlled trajectory of $\mathcal{M}(\hat{f},\hat{G})$ with control signal $u$ where $\phi^{\hat{f},\hat{G}}_u(0;x_0) = x_0$. Then, the \textit{(forward) reachable set}~is $$\mathcal{R}^{\hat{f},\hat{G}}(T, x_0) = \{\phi^{\hat{f},\hat{G}}_u(t; x_0) \hspace{1 mm} | \hspace{1 mm} u : [0,T] \to \mathcal{U}, t\in [0,T]\}.$$
\end{definition}

Using this definition, we can easily define the \textit{guaranteed reachable set} (GRS), which is simply the intersection of all forward reachable sets for all $(\hat{f},\hat{G}) \in \mathcal{D}_{con}$ up to some time~$T$.

\begin{definition}[Guaranteed reachable set]\label{guaranteedReachableSetDefinition}
    Let $T \geq 0$. We describe the \textit{guaranteed reachable set} as:

\begin{equation}\label{guaranteedReachableSet}
    \mathcal{R}^\mathcal{G}(T,x_0) = \bigcap_{(\hat{f},\hat{G}) \in \mathcal{D}_{con}}\mathcal{R}^{\hat{f},\hat{G}}(T,x_0).
\end{equation}
\end{definition}

With the limited knowledge from our assumptions, it is not possible to calculate this set in real time for general cases. This motivates the following problem statement.

\begin{problemStatement}
    Determine or underapproximate the~GRS.
\end{problemStatement}

The GRS can be calculated using the GVS \citep{shafa2022reachability,shafa2022maximal}. Recall that $\mathcal{V}^\mathcal{G}_x$ is the set of all velocities that can be taken by all systems with the assumed knowledge of $\mathcal{D}_{con}$. Let us consider the following ODI:

\begin{equation}\label{ODI}
    \dot{x} \in \mathcal{V}^\mathcal{G}_x,\quad x(0) = x_0.
\end{equation}

\noindent If $\mathcal{V}^\mathcal{G}_{\phi(T;x_0)} = \emptyset$, we will consider by convention that the trajectory of \eqref{ODI} ceases to exist at time $T$. The following proposition then holds directly from \eqref{guaranteedVelocitySet} and \eqref{guaranteedReachableSet} and was directly shown in \citep{ornik2020guaranteed}, just not on $M$.

\begin{proposition}\label{Proposition1}
    Let $T \geq 0$. If a trajectory $\phi:~[0,+\infty) \to M$ satisfies \eqref{ODI} at all times $t \leq T$, then $\phi(T) \in \mathcal{R}^\mathcal{G}(T,x_0)$.
\end{proposition}

Proposition \ref{Proposition1} implies that the reachable set of \eqref{ODI} is a subset of $\mathcal{R}^\mathcal{G}(T,x_0)$. In the next section, we formulate an underapproximation of $\mathcal{V}^\mathcal{G}_x$. We then formulate a control system whose velocities are contained within these underapproximations in Section IV. The union of all trajectories of this control system form the reachable set.
\section{Guaranteed Velocity Set}

In this section, we determine a guaranteed set of attainable velocities the unknown system can reach using knowledge consistent with our assumptions. We begin by determining the domain on $x$ under which Assumption \ref{Assumption 1} remains satisfied, namely that $\mathrm{Im}(\tilde{\tau}_{x_0}^xG(x_0))~=~\mathrm{Im}(G(x))$ using predefined coordinate bases and the flat transport. It is clear from Assumption~\ref{Assumption 1} that $\mathrm{Im}(\tilde{\tau}_{x_0}^xG(x_0))~=~\mathrm{Im}(R)$. We want to quantify some nonzero neighborhood around $x_0$ where $\mathrm{Im}(G(x)) = \mathrm{Im}(R)$ as well.

\begin{lemma}\label{Lemma_Domain}
    Set $g_l^\Gamma = \sum_{i,j,k}\dot{\gamma}^i\Gamma_{ij}^kg_l^j(x_0)\Vec{e}_k$ for $l \in [m]$ where vectors $\Vec{e}_k$ form the basis for $T_xM$. Under Assumptions \ref{Assumption 1} and \ref{Assumption 2}, if 
    \vskip -5pt
    \small

    \begin{equation*}
        d(x_0,x) < \frac{\|\tilde{\tau}_{x_0}^xG^\dagger(x_0)\|^{-1} - \|H^{-1}_x\|^{\frac{1}{2}}\|H_x\|\left\|\begin{bmatrix}g_1^\Gamma & \hdots & g_m^\Gamma\end{bmatrix}\right\|}{(\|H_x^{-1}\|\|H_x\|)^{\frac{1}{2}}L_G},
    \end{equation*}

    \normalsize
    \noindent then $\mathrm{Im}(\tilde{\tau}_{x_0}^xG(x_0)) = \mathrm{Im}(G(x))$.
\end{lemma}

We now compare this result to the case of control systems constrained in flat space.

\begin{remark}\label{Remark_Domain}
    Let $M$ be a flat manifold. Then, the domain which satisfies the condition that $\mathrm{Im}(\tilde{\tau}_{x_0}^xG(x_0)) = \mathrm{Im}(G(x))$ is identical to the domain in Lemma 1 of \citep{shafa2022reachability}.
\end{remark}

The remark above is a consequence of $H_x = I$ and $\Gamma_{ij}^k~=~0$. We emphasize that the remark above holds because the Euclidean space is itself a flat manifold; it also shows that Lemma~1 from \citep{shafa2022reachability} is generalizable to any flat manifold. Now that we have determined the domain under which we will calculate guaranteed velocities, it is time to quantify an underapproximation of the \textit{guaranteed velocity set} (GVS).


\begin{theorem}\label{Theorem 1}
    Let $\hat{f}(x_0),\,\hat{G}(x_0), L_f$, $L_{G}$, $H_x$, $\Gamma_{ij}^k$, and $g_l^\Gamma$ for $l \in [m]$ be defined as above. Let $\gamma~:~[0,1]~\to~M$ define a geodesic curve from $x_0$ to $x$. Let $\Tilde{\tau}$ define the parallel transport using the flat connection. Set $a(x) = (\|H_x^{-1}\|\|H_x\|)^{\frac{1}{2}}\|H_x\|^{\frac{1}{2}}\left\|\begin{bmatrix}\hat{g}_1^\Gamma & \hdots & \hat{g}_m^\Gamma\end{bmatrix}\right\|$, $b(x) = (\|H_x^{-1}\|\|H_x\|)^{\frac{1}{2}}\left\|\sum_{i,j,k}\dot{\gamma}^i\Gamma_{ij}^k\hat{f}^j(x_0)\Vec{e}_k\right\|$, $c(x) = (\|H_x^{-1}\|\|H_x\|)^{\frac{1}{2}}\left(L_G + \|H_x\|^{-\frac{1}{2}}L_f\right)$ and $$\overline{d}(x_0,x) = \frac{\|\tilde{\tau}_{x_0}^x\hat{G}^\dagger(x_0)\|^{-1} - a(x) - b(x)}{c(x)}.$$ If $d(x_0,x) \leq \overline{d}(x_0,x)$,

\begin{equation}\label{Theorem1_Equation}
\overline{\mathcal{V}}^\mathcal{G}_x = \mathbb{B}^n\left(\Tilde{\tau}_{x_0}^x\hat{f}(x_0);\alpha(x_0,x)\right) \cap \mathrm{Im}(\tilde{\tau}_{x_0}^x\hat{G}(x_0))
\end{equation}

\noindent where $\overline{\mathcal{V}}_x^\mathcal{G} \subset T_xM$, and

\vskip -10pt

\begin{equation}\label{Thm1_UpperBound}
\begin{gathered}
    \alpha(x_0,x) = \|\Tilde{\tau}^x_{x_0}\hat{G}^\dagger(x_0)\|^{-1} - \\ (\|H_x^{-1}\|\|H_x\|)^{\frac{1}{2}}\left(\|H_x\|^{\frac{1}{2}}\left\|\begin{bmatrix}\hat{g}_1^\Gamma & \hdots & \hat{g}_m^\Gamma\end{bmatrix}\right\|\right. + \\
    \left.\parenthnewln{} \left\|\sum_{i,j,k}\dot{\gamma}^i\Gamma_{ij}^k\hat{f}^j(x_0)\Vec{e}_k\right\| + \left(L_G + \|H_x\|^{-\frac{1}{2}}L_f\right)d(x_0,x)\right),
\end{gathered}
\end{equation}

\noindent then $\overline{\mathcal{V}}^\mathcal{G}_x \subseteq \mathcal{V}^\mathcal{G}_x$.
\end{theorem}

Note that through an abuse of notation, $\mathbb{B}^n$ in \eqref{Theorem1_Equation} exists in $T_xM$ but is defined in $\mathbb{R}^n$. However, $T_xM$ is diffeomorphic to $\mathbb{R}^n$, so the same discussion below Assumption \ref{Assumption 1} regarding $\mathrm{Im}(\cdot)$ applies here as well. Also, the domain under which we can calculate $\overline{\mathcal{V}}^\mathcal{G}_x$ is contained within the domain derived in Lemma \ref{Lemma_Domain} where $\mathrm{Im}(\tilde{\tau}_{x_0}^xG(x_0)) = \mathrm{Im}(G(x))$, i.e., where assumptions \ref{Assumption 1} and \ref{Assumption 2} hold. In \citep{shafa2022reachability}, we derived a geometrically similar ball underapproximation for systems constrained on Euclidean manifolds. For comparison, we provide this quick remark.

\begin{remark}\label{Remark 1}
    Let $M$ be a flat manifold. Then $\overline{\mathcal{V}}_x^\mathcal{G}$ is equal to the set calculated in Theorem 1 of \citep{shafa2022reachability}.
\end{remark}

Remark~\ref{Remark 1} states that the underapproximation from Theorem 1 of \citep{shafa2022reachability} should hold for a nonlinear control-affine system operating on any flat manifold. Much like Remark~\ref{Remark_Domain}, Remark~\ref{Remark 1} follows because $H_x = I$ and $\Gamma_{ij}^k = 0$ for flat manifolds. Now that we have guaranteed underapproximations of the GVS, we describe in the next section how to use such sets to determine an underapproximation of the GRS which lies on a complete Riemannian manifold.
\section{Reachable Set}


To calculate the guaranteed reachable set of an unknown system, we first use Theorem \ref{Theorem 1} from Section III to determine the guaranteed velocity sets for different $x \in M$. Then, similar to \citep{shafa2022reachability,shafa2022maximal}, we leverage these sets to find a control system whose reachable set is an underapproximation of the GRS. 

We first define an ordinary differential inclusion based on the underapproximation of the GVS in Theorem \ref{Theorem 1}. Then, we formulate a control system whose set of velocities are equal the right hand side of the ODI. The reachable set of this formulated control systems contained in the GRS by Proposition~\ref{Proposition1}. Theorem \ref{Theorem 1} and Proposition~\ref{Proposition1} show that the reachable set of 

\vskip -5pt

\begin{equation}\label{Thm1_ODI}
    \dot{x} \in \overline{\mathcal{V}}^\mathcal{G}_x,\quad x(0) = x_0
\end{equation}

\noindent is a subset of $\mathcal{R}^\mathcal{G}_x(T,x_0)$.

We follow similar steps as in \citep{shafa2022reachability,shafa2022maximal}: Let $\overline{d}(x_0,x)$ and $\alpha(x_0,x)$ be as defined in Theorem \ref{Theorem 1}. Analogous to the interpretation of the dynamics \eqref{nonlinearSystem} as an ODI \eqref{availableVelocitySet}, we can interpret \eqref{Thm1_ODI} as a control system

\vskip -5pt

\begin{equation}\label{ControlSystem1_GRS}
    \dot{x} = a + g(x_0,x)u,\quad x(0) = x_0,
\end{equation}

\noindent on $\{x~|~d(x_0,x) \leq \overline{d}(x_0,x)\}$, with $a = \Tilde{\tau}_{x_0}^xf(x_0)$, $u~\in~\mathbb{B}^n(0;1)$, and $g(s_0,s) = \alpha(x_0,x)$ if $d(s_0,s) \leq \overline{d}(x_0,x)$. We thus obtain the following result.

\begin{theorem}\label{Theorem_GRS1}
    Let $\overline{\mathcal{R}}(T,x_0)$ be defined as the reachable set of \eqref{ControlSystem1_GRS} at time $T$. Then, $\overline{\mathcal{R}}(T,x_0) \subset \mathcal{R}^{\mathcal{G}}(T,x_0)$.
\end{theorem}
\begin{proof}
    Proposition \ref{Proposition1} and Theorem \ref{Theorem 1} prove the claim.
\end{proof}

Now that we have a control system whose solutions are contained in the guaranteed reachable set, we want to determine these solutions numerically on a manifold. Existing numerical methods such as multistep algorithms and sympletic integration techniques like Runge-Kutta \citep{crouch1993numerical,budd1999geometric} for manifolds can be used to calculate a reachable set. We emphasize that systems are often described by ordinary differential equations in Euclidean spaces which constrain their dynamics to a naturally embedded manifold --- see section I.1, I.2, and I.3 in \citep{hairer2011solving} --- or the use of differential algebraic constraint equations \citep{cross2008level}. The solutions to such control systems are computed numerically on Euclidean spaces and projected onto the manifolds. 

We employ the techniques mentioned above to numerically determine the solution to our derived control system whose reachable set is contained in the GRS. For the aforementioned projection methods, numerically solving for the solution to a system evolving on a manifold requires the extra step of mapping velocities in the tangent space to their corresponding trajectories on $M$ via the \textbf{exponential map} \citep{jost2008riemannian,absil2008optimization}. We formally define the exponential map below.

\begin{definition}[(Geodesic) exponential map]
    Let $M$ be a Riemannian manifold and $x \in M$. The exponential map $$\mathrm{exp}_{x_0}:~U \subseteq T_{x_0}M \to M:~\dot{x} \to x(1)$$ at $x_0$ assigns to $\dot{x} \in T_{x_0}M$ the solution at time $t = 1$ to the geodesic equation with initial state $x_0$ and initial velocity $\dot{x}$ where $x(0) = x_0$ and $\nabla_{\dot{x}}\dot{x} = 0$.
\end{definition}

Recall that $M$ is a connected Riemannian manifold. Thus, by the Hopf-Rinow Theorem \citep[Theorem 1.7.1]{jost2008riemannian}, the exponential map is defined on all $T_xM$ for all $x \in M$ and can be leveraged to calculate a set of geodesic curves on $M$ where each curve serves as a viable trajectory contained in the guaranteed reachable set. In the next section, we demonstrate by example how to apply the results in Sections III and IV to calculate the guaranteed reachable set of an unknown system constrained on an abstract manifold. Trajectories of the known system and the guaranteed underapproximation of its reachable set are compared to illustrate the theory's efficacy.
\section{Numerical Examples}

In this section, we demonstrate how to calculate a guaranteed reachable set provably contained within the true reachable set of an unknown system operating operating on $SO(3)$, the space of three-dimensional rotational matrices. The application in question is to determine whether certain rotations can be achieved within a finite time horizon given initial conditions and knowledge from Assumption \ref{Assumption 1} and \ref{Assumption 2}. We emphasize that the structure of $SO(3)$ is significantly different from that of the Euclidean space. Consequently, it is not possible to determine a guaranteed underapproximation for a system that exists on $SO(3)$ using the results in \citep{shafa2022reachability}. 

\subsection{The Set of all 3 x 3 Real Orthogonal Matrices}

\begin{figure*}[ht] 
  \centering
  \begin{minipage}{0.33\textwidth}
    \centering
    \includegraphics[width=\linewidth]{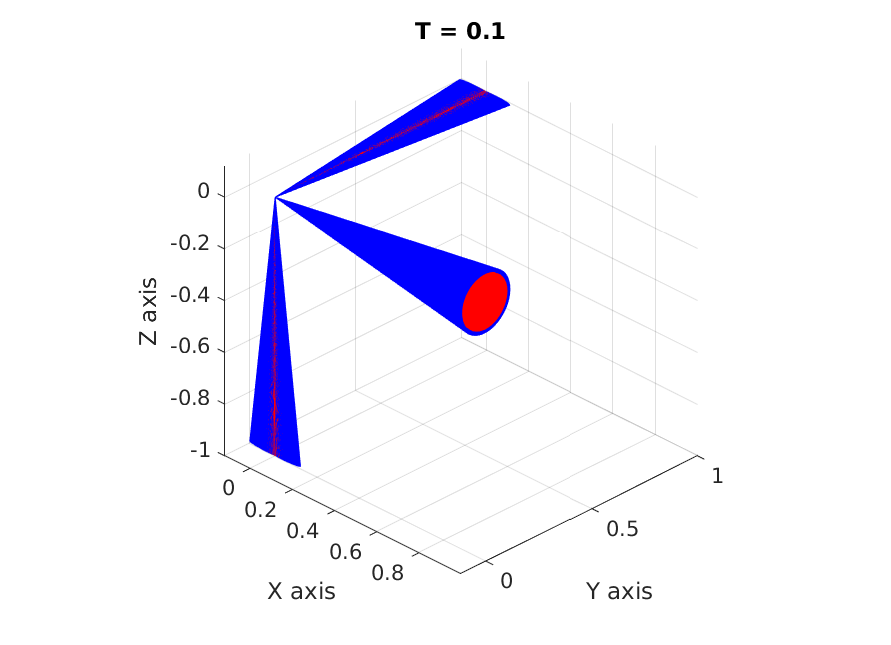} 
    \label{fig:graph1}
  \end{minipage}
  \begin{minipage}{0.33\textwidth}
    \centering
    \includegraphics[width=\linewidth]{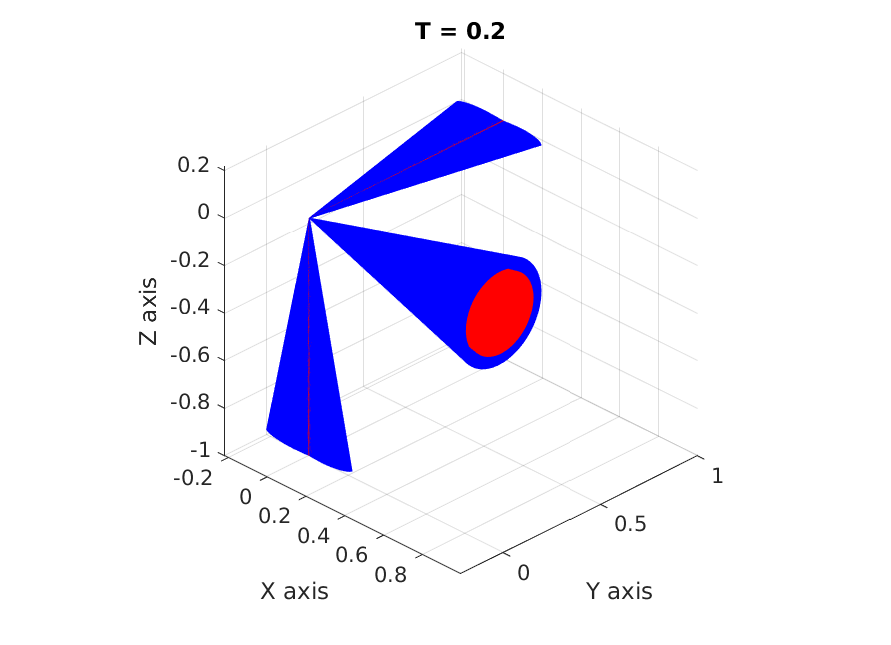} 
    \label{fig:graph2}
  \end{minipage}
  \begin{minipage}{0.33\textwidth}
    \centering
    \includegraphics[width=\linewidth]{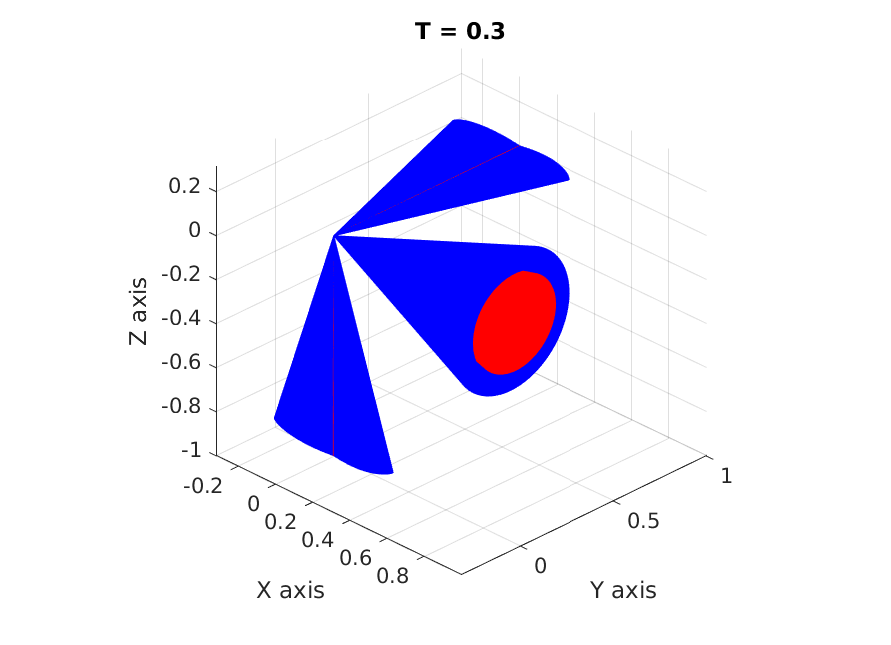} 
    \label{fig:graph3}
  \end{minipage}
    \captionsetup{justification=centering,margin=2cm}
    \caption[font=small]{True reachable set (blue) with the underapproximation $\overline{\mathcal{R}}(T,x_0)$ (red) numerically calculated for $T \in \{0.1, 0.2, 0.3\}$ seconds. Reachable sets display the set of all reachable orthonormal triads with initial condition $X_0$ up to time $T$ using the fully known dynamics \eqref{controlSystem_RotationalMatrices} (blue) and the underapproximated system \eqref{ControlSystem1_GRS} (red).}
    \label{fig:SO(3) Plots}
\end{figure*}

We present an example of a control system which operates on $SO(3) = \{N \in \mathbb{R}^{3\times 3}~|~N^TN = I,\,\mathrm{det}(N) = 1\}$, the special orthogonal group whose elements form the set of all real orthogonal matrices with a determinant of $1$. Let $L(SO(3))$ be the algebra consisting of all $3\times 3$ antisymmetric matrices. A basis for $L(SO(3))$ \citep[Example 8.1]{jurdjevic1972control} is given by the matrices $$K_x = \begin{bmatrix}0 & 0 & 0\\ 0 & 0 & -1\\ 0 & 1 & 0\end{bmatrix},\quad K_y = \begin{bmatrix}0 & 0 & 1\\ 0 & 0 & 0\\ -1 & 0 & 0\end{bmatrix}$$ and $$K_z = \begin{bmatrix}0 & -1 & 0\\ 1 & 0 & 0\\ 0 & 0 & 0\end{bmatrix}.$$ For any point $X \in SO(3)$, we take the basis of $T_XM$ to be $\{K_xX,K_yX,K_zX\}$. We consider $$X(0) = \begin{bmatrix}1 & 0 & 0\\ 0 & 0 & -1\\ 0 & 1 & 0\end{bmatrix}$$ which corresponds to a rotation of $\pi/2$ around the unit vector $(0,1,0)$. To check $X(0) \in SO(3)$, we verify that $\mathrm{det}(X(0)) = 1$ and $X^T(0)X(0) = I$. 

We consider an unknown system evolving on $SO(3)$ with dynamics

\vskip -10pt

\begin{equation}\label{controlSystem_RotationalMatrices}
\begin{gathered}
    \dot{x}(t) = f(x(t)) + G(x(t))u,\quad x(0) = x_0,\\
    \dot{X}(t) = \dot{x}_1K_xX(t) + \dot{x}_2K_yX(t) + \dot{x}_3K_zX(t),
\end{gathered}
\end{equation}

\noindent where $x \in \mathbb{R}^3$, $u \in \mathbb{B}^2(0;1)$, $f = 0$, and $G(x(t))~=~\begin{bmatrix}0 & 0 & 1 + 0.5x_3 \\ 0 & 1 & 0\end{bmatrix}^T$. We operate on the open set $U \in SO(3)$ and $\varphi:U\to \mathbb{R}^3$ such that $(U,\varphi)$ is a chart in $SO(3)$ where $\varphi(U)~=~\{(\psi,\theta,\phi)~:~0<\psi<2\pi,\,0<\theta<\pi,\,0<\phi<2\pi\}$. The initial condition $X(0)$ corresponds to the initial condition $x_0 = (0,\pi/2,0)$. Let $\{e_i:~i=1,\,2,\,3\}$ be an orthonormal triad in $SO(3)$. Each rotation $X$ on $SO(3)$ acts on the triad $\{e_i\}$ to produce another orthonormal triad $\{Xe_i:~i=1,\,2,\,3\}$. Indeed, $X$ is represented under the basis $\{e_i\}$ by the matrix $\begin{bmatrix}Xe_1 & Xe_2 & Xe_3\end{bmatrix}$ with $Xe_i$ as the $i$th column \citep[Section 1.7]{man2022crystallographic}. Local coordinates $x \in \mathbb{R}^3$ represent the Euler angles $(\psi,\theta,\phi)$ which describe the orientation of the orthonormal triad in $SO(3)$.

There are three axes around which $X(t) \in SO(3)$ can rotate, hence this is a $3$-degree of freedom system. We have $\mathcal{U} = \mathbb{B}^2(0;1)$ with $f = 0$, so we can conclude that this is an underactuated control system consistent with Assumption~\ref{Assumption 1}. We can thus calculate the guaranteed set of reachable states using the knowledge from Assumption~\ref{Assumption 2} without knowledge of the true dynamics using the novel theory developed in previous sections. To determine an underapproximation of the GVS, we want to calculate 

\vskip -5pt

\begin{equation}\label{Ex2_eq3}
    \overline{\mathcal{V}}^\mathcal{G}_x = \mathbb{B}^{3}(0;\alpha(x_0,x)) \cap \mathrm{Im}(\tilde{\tau}_{x_0}^xG(x_0)),
\end{equation}

\noindent so we need to solve for $\alpha(x_0,x)$. It follows that we need to quantify $\tilde{\tau}_{x_0}^xf(x_0)$,  $\|\tilde{\tau}_{x_0}^xG^\dagger(x_0)\|^{-1}$, $d(x_0,x)$, $L_f$, $L_G$, $\|H_x^{-1}\|^{-1}$, $\|H_x\|$, and $\Gamma_{ij}^k(x)$ to calculate the term $\alpha(x_0,x)$ defined in \eqref{Thm1_UpperBound} for any $X~\in~SO(3)$. By Assumption \ref{Assumption 2}, we know $f(x_0) = 0$ and so trivially $\tilde{\tau}_{x_0}^xf(x_0)~=~0$ and $L_f = 0$. Since initial conditions $(x_{1}(0),x_2(0),x_3(0)) = (\psi_0,\theta_0,\phi_0) = (0,\pi/2,0)$, it follows that $\|\tilde{\tau}_{x_0}^xG^\dagger(x_0)\|^{-1} = 1$. Note that by metric compatibility, $\|\tilde{\tau}_{x_0}^xG^\dagger(x_0)\|^{-1}$ is invariant under the transport $\tilde{\tau}$. 

On the open set $U$ defined earlier, $d(x,x_0)$ using the Euler angles $(\psi,\theta,\phi)$ is defined as the misorientation angle --- the angle between two elements $X_0,X \in SO(3)$. From \citep[Section 10.2]{man2022crystallographic} we can define the misorientation of $X$ with respect to $X_0$ given below $$d(x_0,x) = \cos^{-1}\left({\frac{1}{2}\mathrm{Tr}(X_0X^T) - 0.5}\right).$$ In \citep[Section 10.3]{man2022crystallographic}, for the neighborhood $U$ we calculate the metric tensor and inverse metric tensor as $$H_x = \begin{bmatrix}1 & 0 & \cos{\theta}\\ 0 & 1 & 0\\ \cos{\theta} & 0 & 1\end{bmatrix},\,H^{-1}_x = \begin{bmatrix}
    \frac{1}{\sin^2{\theta}} & 0 & -\frac{cos{\theta}}{\sin^2{\theta}}\\ 0 & 1 & 0\\ -\frac{cos{\theta}}{\sin^2{\theta}} & 0 & \frac{1}{\sin^2{\theta}}
\end{bmatrix}.$$ With $H_x$ and $H^{-1}_x$ we can calculate $\Gamma_{ij}^k(x)$ using equation \eqref{connectionCoefficients}. Lastly, for some neighborhood $V \subset U$ on which we operate, by Assumption \ref{Assumption 2} we are given $L_{g_1}~=~0.65$, $L_{g_2}~=~0$, $\|H_x\|~<~1.2$, $\|H_x^{-1}\|^{-1}~>~0.8$ for all $x \in V$. From Lemma \ref{Lemma_LGBound}, we calculate a valid upper bound for equation \eqref{LGBound} as $L_G~=~n\|H_x^{-1}\|\|H_x\|^{\frac{1}{2}}L_{g_1} \approx 2.7$.


The corresponding GRS calculated using $\mathcal{V}^\mathcal{G}_x$ from \eqref{Ex2_eq3} is shown in Fig. \ref{fig:SO(3) Plots} for $T \in \{0.1, 0.2, 0.3\}$. Fig. \ref{fig:SO(3) Plots} shows all possible orientations of an orthonormal triad with initial condition $X_0$ within a finite time horizon for the known control system \eqref{controlSystem_RotationalMatrices} and the underapproximated control system \eqref{ControlSystem1_GRS}. The blue represents all possible orientations using the known control system \eqref{controlSystem_RotationalMatrices} while the red represent the same information gathered using solely the novel theory presented in this paper without knowledge of \eqref{controlSystem_RotationalMatrices}. Consistent with the theory presented in this paper, the GRS calculated using the presented theory is contained within the interior of the union of all blue triads which represent the reachable set calculated using the known dynamics of system \eqref{controlSystem_RotationalMatrices}.
\section{Conclusion and Future Work}

This paper presents a novel approach for determining an underapproximation of an unknown nonlinear control system's reachable set which lies on any complete Riemannian manifold. By assuming a nonlinear control-affine structure and knowledge of the dynamics at a single point, we can produce a set $\overline{\mathcal{R}}(T,x_0) \subseteq \mathcal{R}^\mathcal{G}(T,x_0)$ which consists of provably reachable trajectories for the unknown system. The underapproximation relies on an intermediate approximation of the guaranteed set of reachable states by an ODI $\dot{x} \in \mathcal{V}^\mathcal{G}_x$ where its right-hand side is a set of guaranteed velocities for the unknown nonlinear system. We demonstrate the efficacy of our approach through application on a system that lies on the abstract manifold of real orthogonal matrices $SO(3)$ whose structure is not diffeomorphic to the Euclidean space.

A natural area of future work is to increase the knowledge of the system dynamics beyond Riemannian Lipschitz bounds, i.e., to utilize additional information to increase the size of our underapproximations without knowledge of the true system dynamics. Instead of maximizing the knowledge from just one trajectory, we can incorporate the knowledge from many trajectories. This can potentially generalize results to include all underactuated nonlinear control-affine systems by altering the space in which the unknown dynamics lie. Another natural extension of this research is to extend the results beyond the deterministic domain and consider how similar methods can apply to models constructed by neural nets with parametric noise. This noise would pertain to modeling uncertainty, possibly due to to error in system identification for unknown systems using multiply trajectories. Such efforts would provide certifiable capabilities for neural network system models which often lack robust provable performance guarantees. 
\section{Acknowledgements}

We thank Yiming Meng from the University of Illinois Urbana-Champaign for helping revise the paper to make concepts related to advanced geometry more clear to controls engineers and theorists without a rigorous geometric background. 


\bibliographystyle{IEEEtran}
\bibliography{root.bib}

\section{Appendix}
\subsection{Proofs of Supporting Lemmata}

\noindent \begin{proof}[Proof of Lemma \ref{Lemma_normEquivalenceRelations}]
The matrix $H_x$ is diagonalizable by the symmetry of $H_x$ \citep{strang2016introduction}, so by \eqref{riemannianNorm} we have $|v|_{h_x}^2 = \sum_i\lambda_i(H_x)(v^i)^2 \leq \lambda_{\mathrm{max}}(H_x)\|v\|^2$. For symmetric matrices, $\lambda_i(H_x) = \sigma_i(H_x)$ for all $i$ \citep{strang2016introduction}, hence $\lambda_{\mathrm{max}}(H_x) = \|H_x\|$. Taking the square root of both sides gives us the right-hand side of \eqref{vectorNormEquivalency}. Repeating a similar process we see that $\lambda_{\mathrm{min}}(H_x)\|v\|^2 \leq |v|_{h_x}^2$ implies the left-hand inequality in \eqref{vectorNormEquivalency}.

    Equations \eqref{riemannianMatrixNorm} and \eqref{vectorNormEquivalency} imply $$|A|_{h_x} \leq \sup_{v\neq 0}\frac{\|H_x\|^{\frac{1}{2}}\|Av\|}{\|H^{-1}_x\|^{-\frac{1}{2}}\|v\|} = \frac{\|H_x\|^{\frac{1}{2}}}{\|H^{-1}_x\|^{-\frac{1}{2}}}\sup_{v \neq 0}\frac{\|Av\|}{\|v\|}.$$ By definition of the Euclidean matrix norm \citep{strang2016introduction} and the inequality above, we get the right-hand side of \eqref{matrixNormEquivalency}. Similar steps using the lower bound in \eqref{vectorNormEquivalency} gives us the left-hand side of \eqref{matrixNormEquivalency}.
\end{proof}

\noindent\begin{proof}[Proof of Lemma \ref{Lemma_LGBound}]
We begin by setting $l \in [m]$ to the index such that $L_{g_l}~=~\overline{L}_{g}$. By the equivalence relation \eqref{vectorNormEquivalency}, we have $\|H_x^{-1}\|^{-\frac{1}{2}}\|\tau_{x_0}^xg_l(x_0) - g_l(x)\| \leq |\tau_{x_0}^xg_l(x_0) - g_l(x)|_{h_x}$. Additionally, through vector norm equivalence relations \citep{strang2016introduction}, $\|\tau_{x_0}^xg_l(x_0) - g_l(x)\|_1 \leq \sqrt{n}\|\tau_{x_0}^xg_l(x_0) - g_l(x)\|$. Using the bound from Assumption \ref{Assumption 2}, we can thus conclude that 

    \begin{equation}\label{Lemma_LGBound_eq1}
        \frac{\|H_{x}^{-1}\|^{-\frac{1}{2}}}{\sqrt{n}}\|\tau_{x_0}^xg_l(x_0) - g(x)\|_1 \leq \overline{L}_gd(x_0,x).
    \end{equation}
    
    \noindent Notice that, by the definition of the induced matrix $1$-norm \citep{strang2016introduction}, $\|\tau_{x_0}^xg_l(x_0) - g(x)\|_1 = \|\tau_{x_0}^xG(x_0) - G(x)\|_1$. By matrix norm equivalence relations \citep{strang2016introduction}, $\|\tau_{x_0}^xG(x_0) - G(x)\| \leq \sqrt{n}\|\tau_{x_0}^xG(x_0) - G(x)\|_1$. Using \eqref{matrixNormEquivalency}, we also get $(\|H_x\|\|H_x^{-1}\|)^{-\frac{1}{2}}|\tau_{x_0}^xG(x_0) - G(x)|_{h_x} \leq \|\tau_{x_0}^xG(x_0) - G(x)\|$. This inequality implies $$\frac{(\|H_x\|\|H_x^{-1}\|)^{-\frac{1}{2}}}{\sqrt{n}}|\tau_{x_0}^xG(x_0) - G(x)|_{h_x} \leq \|\tau_{x_0}^xg_l(x_0) - g(x)\|_1.$$ Substituting the left-hand side of this inequality into the $\|\tau_{x_0}^xg_l(x_0) - g(x)\|_1$ term in \eqref{Lemma_LGBound_eq1} and multiplying both sides by $n\|H_x^{-1}\|\|H_x\|^{\frac{1}{2}}$ proves the claim.
\end{proof}

\noindent\begin{proof}[Proof of Lemma \ref{Lemma_Domain}]

To prove the claim, we utilize Weyl's inequality \citep{stewart1998perturbation} for singular values, which provides a Euclidean-normed bound on the perturbations of singular values. Consequently, we are motivated to determine an upper bound on $\|\Tilde{\tau}_{x_0}^xG(x_0) - G(x)\|$ to exploit metric compatibility properties between the flat connection and its corresponding Euclidean norm. We note that $\|\Tilde{\tau}_{x_0}^xG(x_0) - G(x)\| = \|\Tilde{\tau}_{x_0}^xG(x_0) - \tau_{x_0}^xG(x_0) + \tau_{x_0}^xG(x_0) - G(x)\| \leq \|\Tilde{\tau}_{x_0}^xG(x_0) - \tau_{x_0}^xG(x_0)\| + \|\tau_{x_0}^xG(x_0) - G(x)\|$, so any upper bound on $\|\Tilde{\tau}_{x_0}^xG(x_0) - \tau_{x_0}^xG(x_0)\| + \|\tau_{x_0}^xG(x_0) - G(x)\|$ is also an upper bound on $\|\Tilde{\tau}_{x_0}^xG(x_0) - G(x)\|$. 
    
    We begin by finding an upper bound on $\|\tau_{x_0}^xG(x_0) - G(x)\|$ using inequality \eqref{LGBound} from Lemma \ref{Lemma_LGBound} and substituting the left-hand side of inequality \eqref{matrixNormEquivalency} into the left-hand side of \eqref{LGBound}. We now multiply both sides by $(\|H_x^{-1}\|\|H_x\|)^{\frac{1}{2}}$ to get
    \vskip -10pt
    \begin{equation}\label{Lemma_Domain_eq1}
        \|\tau_{x_0}^xG(x_0) - G(x)\| \leq (\|H_x^{-1}\|\|H_x\|)^{\frac{1}{2}}L_Gd(x_0,x).
    \end{equation}

    Next, we need to quantify $|\tilde{\tau}_{x_0}^xG(x_0) - \tau_{x_0}^xG(x_0)|_{h_x}$ and use the norm equivalency relations from Lemma \ref{Lemma_normEquivalenceRelations} to find an upper bound on $\|\tilde{\tau}_{x_0}^xG(x_0) - \tau_{x_0}^xG(x_0)\|$. By the definitions of the affine connection \citep[Section 9.1]{frankel2011geometry} and and parallel transport \citep[Definition 4.1.2]{jost2008riemannian}, we find that if we travel along some curve $\gamma:[0,1] \to M$, the parallel transport provides a unique parallel vector field in terms of the covariant derivative corresponding to the appropriate connection. We want to measure the difference of these aforementioned parallel vector fields expressed in terms of the covariant derivative. Thus, we measure $|\tilde{\tau}_{x_0}^xG(x_0)~-~\tau_{x_0}^xG(x_0)|_{h_x}$ by quantifying $|\tilde{\nabla}_{\dot{\gamma}}G(x_0)~-~\nabla_{\dot{\gamma}}G(x_0)|_{h_x}$ where $\tilde{\nabla},\nabla$ are the connections corresponding to $\tilde{\tau}$ and $\tau$ respectively. We then convert to the desired Euclidean metric to apply Weyl's inequality.

    We begin by reminding the reader that $G(x_0)$ is defined $G(x_0) = \begin{bmatrix}g_1(x_0) & \hdots & g_m(x_0)\end{bmatrix}$ such that $g_l(x_0) \in T_{x_0}M$. In other words, $G(x_0)$ is treated as a collection matrix of vectors in $T_{x_0}M$, so to parallel transport $G(x_0)$ is to perform said transport $\tau_{x_0}^xg_l(x_0)$ for all $l \in [m]$. By extension, $\nabla_{\dot{\gamma}}G(x_0) = \begin{bmatrix}\nabla_{\dot{\gamma}}g_1(x_0) & \hdots & \nabla_{\dot{\gamma}}g_m(x_0)\end{bmatrix}$ is also defined component-wise for each component vector $g_l(x_0)$.

    For any $l \in [m]$, using \eqref{covariantDerivative_Definition} and substituting $\Tilde{\Gamma}_{ij}^k = 0$ for all $i,j,k$, we can calculate 

    \begin{equation}\label{Lemma_Domain_eq2}\tilde{\nabla}_{\dot{\gamma}}g_l(x_0) - \nabla_{\dot{\gamma}}g_l(x_0) = \sum_{i,j,k}\dot{\gamma}^i\Gamma_{ij}^k(x)g_l^j(x_0)\Vec{e}_k.
    \end{equation}

    \noindent For notational simplicity, let $g_l^\Gamma$ equal the right-hand side of \eqref{Lemma_Domain_eq2} for any $l \in [m]$. The unique vectors that results from flat and parallel transport satisfy $\tilde{\nabla}_{\dot{\gamma}}g_l(x_0)~=~0$ and $\nabla_{\dot{\gamma}}g_l(x_0)~=~0$ \citep[Definition 4.1.2]{jost2008riemannian} respectively. That is, the right-hand side of \eqref{Lemma_Domain_eq2} equals $0$. With this and $\tilde{\nabla}_{\dot{\gamma}}G(x_0) - \nabla_{\dot{\gamma}}G(x_0) = \begin{bmatrix}g_1^\Gamma & \hdots & g_m^\Gamma\end{bmatrix}$, we conclude
    \begin{equation}\label{Lemma_Domain_eq3}
    \left|\tilde{\tau}_{x_0}^xG(x_0) - \tau_{x_0}^xG(x_0)\right|_{h_x} = \left|\begin{bmatrix}g_1^\Gamma & \hdots & g_m^\Gamma\end{bmatrix}\right|_{h_x}.
    \end{equation}
    
    \noindent The left and right-hand sides of \eqref{Lemma_Domain_eq3} can be bounded by the left and right-hand sides of inequalities \eqref{vectorNormEquivalency} and \eqref{matrixNormEquivalency} respectively. Then we multiply both sides by $(\|H_x^{-1}\|\|H_x\|)^{\frac{1}{2}}$ to get 

    \begin{equation}\label{Lemma_Domain_eq4}
        \|\tilde{\tau}_{x_0}^xG(x_0) - \tau_{x_0}^xG(x_0)\| \leq \|H_x^{-1}\|^{\frac{1}{2}}\|H_x\|\left\|\begin{bmatrix}g_1^\Gamma & \hdots & g_m^\Gamma\end{bmatrix}\right\|.
    \end{equation}

    \noindent From \eqref{Lemma_Domain_eq1} and \eqref{Lemma_Domain_eq4} it follows that 

    \begin{equation}\label{Lemma_Domain_eq5}
    \begin{gathered}
        \|\tilde{\tau}_{x_0}^xG(x_0) - G(x)\| \leq \\
        (\|H_x^{-1}\|\|H_x\|)^{\frac{1}{2}}\left(\|H_x\|^{\frac{1}{2}} \left\|\begin{bmatrix}g_1^\Gamma & \hdots & g_m^\Gamma\end{bmatrix}\right\|+ L_Gd(x_0,x)\right).
    \end{gathered}
    \end{equation}
    
    For notational ease, let us denote $\eta(x_0,x)$ as the right hand side of \eqref{Lemma_Domain_eq5}. As discussed earlier, the flat transport is said to be metric compatible with the Euclidean norm, i.e., $\|G(x_0)\| = \|\tilde{\tau}_{x_0}^xG(x_0)\|$ because $\Gamma_{ij}^k = 0$ for the flat connection and its corresponding Euclidean metric. We can now apply Weyl's inequality, which states $\|\sigma_s(G(x)) - \sigma_s(\tilde{\tau}_{x_0}^xG(x_0))\| \leq \|G(x) - \tilde{\tau}_{x_0}^xG(x_0)\| \leq \eta(x_0,x)$ such that $1 \leq s \leq r$ with $r = \mathrm{rank}(\tilde{\tau}_{x_0}^xG(x_0))$. By the Eckhart-Young Mirsky theorem \citep{strang2016introduction} and singular value decomposition, it holds that $\sigma_r(G(x_0)) = \|G^\dagger(x_0)\|^{-1} = \|\tilde{\tau}_{x_0}^xG^\dagger(x_0)\|^{-1}$; the second equality stems from the norm preservation property of compatible metrics. Note that $\|\tilde{\tau}_{x_0}^xG^\dagger(x_0)\|^{-1}$ is equal to the smallest nonzero singular value of $\tilde{\tau}_{x_0}^xG(x_0)$. Let $d(x_0,x)$ satisfy
    
    \vskip -10pt
    \small
    \begin{equation}\label{Lemma_Domain_eq6}
        d(x_0,x) < \frac{\|\tilde{\tau}_{x_0}^xG^\dagger(x_0)\|^{-1} - \|H^{-1}_x\|^{\frac{1}{2}}\|H_x\|\left\|\begin{bmatrix}g_1^\Gamma & \hdots & g_m^\Gamma\end{bmatrix}\right\|}{(\|H_x^{-1}\|\|H_x\|)^{\frac{1}{2}}L_G}.
    \end{equation}
    \normalsize
    
    \noindent Substituting the right-hand side of \eqref{Lemma_Domain_eq6} into \eqref{Lemma_Domain_eq5} we get $\|\tilde{\tau}_{x_0}^xG(x_0) - G(x)\| \leq \|\tilde{\tau}_{x_0}^xG^\dagger(x_0)\|^{-1} = \sigma_r(\tilde{\tau}_{x_0}^xG(x_0))$. Thus, by Weyl's inequality we have $\|\sigma_s(G(x)) - \sigma_s(\tilde{\tau}_{x_0}^xG(x_0))\| < \sigma_r(\tilde{\tau}_{x_0}^xG(x_0))$, i.e., $\sigma_s(G(x)) > 0$. Hence, $\mathrm{rank}(G(x)) \geq \mathrm{rank}(\tilde{\tau}_{x_0}^xG(x_0))$. 

    We now show that for any $d(x_0,x)$ which satisfies \eqref{Lemma_Domain_eq6}, $\mathrm{rank}(G(x)) = \mathrm{rank}(\tilde{\tau}_{x_0}^xG(x_0))$. If we substitute the left-hand side of \eqref{Lemma_Domain_eq4} into the right-hand side of \eqref{Lemma_Domain_eq6}, we get $$d(x_0,x) < \frac{\|\tilde{\tau}_{x_0}^xG^\dagger(x_0)\|^{-1} - \|\tilde{\tau}_{x_0}^xG(x_0) - \tau_{x_0}^xG(x)\|}{(\|H_x^{-1}\|\|H_x\|)^{\frac{1}{2}}L_G}.$$ This inequality implies $$\|\tilde{\tau}_{x_0}^xG(x_0) - \tau_{x_0}^xG(x)\| < $$ $$\|\tilde{\tau}_{x_0}^xG^\dagger(x_0)\|^{-1} - (\|H_x^{-1}\|\|H_x\|)^\frac{1}{2}L_Gd(x_0,x).$$ The term $(\|H_x^{-1}\|\|H_x\|)^\frac{1}{2}L_Gd(x_0,x)$ is nonnegative, thus $\|\tilde{\tau}_{x_0}^xG(x_0) - \tau_{x_0}^xG(x)\| < \|\tilde{\tau}_{x_0}^xG^\dagger(x_0)\|^{-1}$ holds as well. Following similar steps as above, we see that $\mathrm{rank}(\tilde{\tau}_{x_0}^xG(x_0)) \geq \mathrm{rank}(\tau_{x_0}^xG(x_0))$. Recall we defined $G(x) = RH(x)$ by Assumption \ref{Assumption 1}, so $\mathrm{Im}(G(x)) \subset \mathrm{Im}(R)$. Since $\mathrm{Im}(G(x_0)) = \mathrm{Im}(R)$, then $\mathrm{Im}(G(x)) \subset \mathrm{Im}(G(x_0))$. Thus, $\mathrm{rank}(G(x_0)) \geq 
    \mathrm{rank}(G(x))$.  On any complete Riemannian manifold, $\mathrm{rank}(G(x_0)) = \mathrm{rank}(\tau_{x_0}^xG(x_0))$. So we have $\mathrm{rank}(\Tilde{\tau}_{x_0}^xG(x_0)) \geq \mathrm{rank}(\tau_{x_0}^xG(x_0)) = \mathrm{rank}(G(x_0)) \geq \mathrm{rank}(G(x))$. Hence, $\mathrm{rank}(\Tilde{\tau}_{x_0}^xG(x_0)) = \mathrm{rank}(G(x))$. By Assumption \ref{Assumption 1}, we now conclude that $\mathrm{Im}(\Tilde{\tau}_{x_0}^xG(x_0)) = \mathrm{Im}(G(x)) = \mathrm{Im}(R)$ when \eqref{Lemma_Domain_eq6} is satisfied.
\end{proof}

\subsection{Proof of Theorem \ref{Theorem 1}}

\noindent \begin{proof}[Proof of Theorem \ref{Theorem 1}]

Let $d \in \mathrm{Im}(\tilde{\tau}_{x_0}^xG(x_0)) = \mathrm{Im}(R)$ such that $\|d\|~=~1$. We will prove that if $|k| \leq \alpha(x_0,x)$, then 

\begin{equation}\label{Thm1_eq0}
    kd + \tilde{\tau}_{x_0}^xf(x_0) = f(x) + G(x)u
\end{equation}

\noindent admits a solution $u \in \mathcal{U} = \mathbb{B}^m(0;1)$ for any $(f,G) \in \mathcal{D}_{con}$. If we let $\eta(x_0,x)$ equal the right hand side of \eqref{Lemma_Domain_eq6}, we see that $\overline{d}(x_0,x) \leq \eta(x_0,x)$. Hence, by Assumption \ref{Assumption 1} and Lemma \ref{Lemma_Domain}, $\mathrm{Im}(G(x)) = \mathrm{Im}(R)$. We subtract both sides by $f(x) \in \mathrm{Im}(R)$. With $kd + \tilde{\tau}_{x_0}^xf(x_0) \in \mathrm{Im}(R)$, we have $kd + \tilde{\tau}_{x_0}^xf(x_0) - f(x) \in \mathrm{Im}(R)$. Hence, there exists a vector $\bar{u} \in \mathbb{R}^m$ such that $kd + \tilde{\tau}_{x_0}^xf(x_0) - f(x) = G(x)\bar{u}$. From here, we consider three cases, namely when $m < n$, $m = n$, and $m > n$.

We begin with the case when $m < n$. Through the rank nullity theorem \citep{strang2016introduction}, we can write $\bar{u} = u + u_2$ where $u \in \mathrm{Im}(G^T(x))$ and $u_2 \in \mathrm{Ker}(G(x))$. Thus, $G(x)\bar{u} = G(x)(u + u_2) = G(x)u$, and so $kd + \tilde{\tau}_{x_0}^xf(x_0) - f(x) = G(x)u$. We multiply both sides of $kd + \tilde{\tau}_{x_0}^xf(x_0) - f(x) = G(x)u$ on the left by $G^\dagger(x)$ and get $G^\dagger(x)(kd + \tilde{\tau}_{x_0}^xf(x_0) - f(x)) = G^\dagger(x)G(x)u$. The term $G^\dagger(x)G(x)u$ is the projection of $u$ onto $\mathrm{Im}(G^T(x))$ \citep{strang2016introduction}. Given that $u \in \mathrm{Im}(G^T(x))$, by definition of a projection, $G^\dagger(x)(kd + \tilde{\tau}_{x_0}^xf(x_0) - f(x)) = G^\dagger(x)G(x)u = u$. Taking the Euclidean norm of both sides, we show that if

\begin{equation}\label{Thm1_eq1}
    \left\|G^\dagger(x)(kd + \tilde{\tau}_{x_0}^xf(x_0) - f(x))\right\| \leq 1,
\end{equation}

\noindent then $\|u\| \leq 1$, i.e., $u \in \mathcal{U}$. For the instances where $m = n$ and $m > n$, $\mathrm{span}\{G(x)\} = T_xM$, so we can trivially follow a similar procedure to conclude that if \eqref{Thm1_eq1} is satisfied, there exists at least one $u \in \mathcal{U}$ which serves as a solution to \eqref{Thm1_eq0}.

The product and triangle inequality for norms gives us $\|G^\dagger(x)(kd + \tilde{\tau}_{x_0}^xf(x_0) - f(x))\| \leq \|G^\dagger(x)\|(|k|\|d\| + \|\tilde{\tau}_{x_0}^xf(x_0) - f(x)\|)$. Hence, any $k$ which satisfies 

\begin{equation}\label{Thm1_eq2}
    \left\|G^\dagger(x)\right\|(|k|\|d\| + \left\|\tilde{\tau}_{x_0}^xf(x_0) - f(x)\right\|) \leq 1
\end{equation}

\noindent satisfies \eqref{Thm1_eq1} and admits a solution $u \in \mathcal{U}$. We may substitute $\|d\| = 1$; multiplying on the left by $\|G^\dagger(x)\|^{-1}$ and subtracting $\|\tilde{\tau}_{x_0}^xf(x_0) - f(x)\|$ from both sides, we see that any $k$ which satisfies

\begin{equation}\label{Thm1_eq3}
    |k| \leq \left\|G^\dagger(x)\right\|^{-1} - \left\|\tilde{\tau}_{x_0}^xf(x_0) - f(x)\right\|
\end{equation}

\noindent also satisfies \eqref{Thm1_eq2} and thus admits a solution $u \in \mathcal{U}$.

From \eqref{Lemma_Domain_eq5} in the proof of Lemma \ref{Lemma_Domain} we have an upper bound on $\|G(x) - \tilde{\tau}_{x_0}^xG(x_0)\|$. We then apply the reverse triangle inequality and add $\|\tilde{\tau}_{x_0}^xG(x_0)\|$ to both sides of \eqref{Lemma_Domain_eq5} to get 

\begin{equation}\label{Thm1_eq10}
\begin{gathered}
    \|G(x)\| \leq \|\tilde{\tau}_{x_0}^xG(x_0)\| + \\ \left(\frac{\|H_x\|}{\|H_x^{-1}\|^{-1}}\right)^{\frac{1}{2}}\left(\|H_x\|^{\frac{1}{2}}\left\|\begin{bmatrix}g_1^\Gamma & \hdots & g_m^\Gamma\end{bmatrix}\right\| + L_Gd(x_0,x)\right).
\end{gathered}
\end{equation}

\noindent Note that $\|G^\dagger(x)\|^{-1}$ is equal to the smallest singular value of $G(x)$ and that $\|G^\dagger(x)\|^{-1} = \|\tilde{\tau}_{x_0}^xG^\dagger(x)\|^{-1}$ by metric compatibility. Thus, we can apply Weyl's inequality for singular values \citep{stewart1998perturbation} to get

\vskip -5pt

\begin{equation}\label{Thm1_eq6}
\begin{gathered}
    \|G^\dagger(x)\|^{-1} \geq \|\tilde{\tau}_{x_0}^xG^\dagger(x_0)\|^{-1} - \\ \left(\frac{\|H_x\|}{\|H_x^{-1}\|^{-1}}\right)^{\frac{1}{2}}\left(\|H_x\|^{\frac{1}{2}}\left\|\begin{bmatrix}g_1^\Gamma & \hdots & g_m^\Gamma\end{bmatrix}\right\| + L_Gd(x_0,x)\right).
\end{gathered}
\end{equation} 

\vskip -5pt

We also want to bound $|\tilde{\tau}_{x_0}^xf(x_0) - f(x)|_{h_x}$ in terms we know from our assumptions. It holds that $|\tilde{\tau}_{x_0}^xf(x_0) - f(x)|_{h_x} \leq |\tilde{\tau}_{x_0}^xf(x_0) - \tau_{x_0}^xf(x_0)|_{h_x} + |\tau_{x_0}^xf(x_0) - f(x)|_{h_x}$. From Assumption~\ref{Assumption 2}, we know that $|\tau_{x_0}^xf(x_0) - f(x)|_{h_x} \leq L_fd(x_0,x)$. The left-hand side of inequality \eqref{vectorNormEquivalency} the inequality above imply that $\|\tau_{x_0}^xf(x_0) - f(x)\| \leq \|H_x^{-1}\|^{\frac{1}{2}}L_fd(x_0,x)$. Repeating a similar process as in the proof of Lemma \ref{Lemma_Domain}, we have that $$|\tilde{\tau}_{x_0}^xf(x_0) - \tau_{x_0}^xf(x_0)|_{h_x} \leq \left|\sum_{i,j,k}\dot{\gamma}^i\Gamma_{ij}^kf^j(x_0)\Vec{e}_k\right|_{h_x}$$ where vectors $\Vec{e}_k$ form the basis for $T_xM$. Again, the inequalities in \eqref{vectorNormEquivalency} with the inequality above imply that 
\vskip -5pt
\small

\begin{equation*}
    \|\tilde{\tau}_{x_0}^xf(x_0) - \tau_{x_0}^xf(x_0)\| \leq \|H_x^{-1}\|^{\frac{1}{2}} \|H_x\|^{\frac{1}{2}}\left\|\sum_{i,j,k}\dot{\gamma}^i\Gamma_{ij}^kf^j(x_0)\Vec{e}_k\right\|
\end{equation*}

\normalsize

\noindent holds. With both terms bounded, we can conclude that

\begin{equation}\label{Thm1_eq7}
\begin{gathered}
    \|\tilde{\tau}_{x_0}^xf(x_0) - f(x)\| \leq \\ \|H_x^{-1}\|^{\frac{1}{2}} \left(\|H_x\|^{\frac{1}{2}}\left\|\sum_{i,j,k}\dot{\gamma}^i\Gamma_{ij}^kf^j(x_0)\Vec{e}_k\right\| + L_fd(x_0,x)\right).
\end{gathered}
\end{equation}

\noindent From \eqref{Thm1_eq6} and \eqref{Thm1_eq7}, if $k$ satisfies

\begin{equation}\label{Thm1_eq8}
    |k| \leq \alpha(x_0,x),
\end{equation}

\noindent then it satisfies \eqref{Thm1_eq3} and by extension \eqref{Thm1_eq1}, thereby admitting a solution $u \in \mathcal{U}$ and proving the claim.
\end{proof}
\end{document}